\documentclass[Crown,sagev,times]{sagej}

\usepackage{moreverb,url}

\usepackage[colorlinks,bookmarksopen,bookmarksnumbered,citecolor=red,urlcolor=red]{hyperref}

\newcommand\BibTeX{{\rmfamily B\kern-.05em \textsc{i\kern-.025em b}\kern-.08em
T\kern-.1667em\lower.7ex\hbox{E}\kern-.125emX}}

\def\volumeyear{2016}

\newcommand{\Ex}{\mathbb{E}}
\newcommand{\var}{\textup{Var}}
\newcommand{\se}{\textrm{se}}
\newcommand{\Dcon}{\stackrel{D}{\rightarrow}}
\newcommand{\Pcon}{\stackrel{P}{\rightarrow}}

\newtheorem{theorem}{\underline{\bf Theorem}}

\newtheorem{remark}{\underline{\bf Remark}}

\newtheorem{lemma}{\underline{\bf Lemma}}

\begin{document}

\runninghead{Li et al.}

\title{Testing for treatment effect in covariate-adaptive randomized trials with generalized linear models and omitted covariates}

\author{Yang Li\affilnum{1}, Wei Ma\affilnum{2}, Yichen Qin\affilnum{3} and Feifang Hu \affilnum{4}}

\affiliation{\affilnum{1}Center for Applied Statistics and School of Statistics, Renmin University of China, Beijing, China\\
\affilnum{2}Institute of Statistics and Big Data, Renmin University of China, Beijing, China\\
\affilnum{3}Department of Operations, Business Analytics, and Information Systems, University of Cincinnati, Cincinnati, OH, USA\\
\affilnum{4}Department of Statistics, George Washington University, Washington, DC, USA}

\corrauth{Wei Ma, Institute of Statistics and Big Data, Renmin University of China,
59 Zhongguancun St, Beijing, China}

\email{mawei@ruc.edu.cn}

\begin{abstract}
Concerns have been expressed over the validity of statistical inference under covariate-adaptive randomization despite the extensive use in clinical trials. In the literature, the inferential properties under covariate-adaptive randomization have been mainly studied for continuous responses; in particular, it is well known that the usual two sample t-test for treatment effect is typically conservative. This phenomenon of invalid tests has also been found for generalized linear models without adjusting for the covariates and are sometimes more worrisome due to inflated Type I error. The purpose of this study is to examine the unadjusted test for treatment effect under generalized linear models and covariate-adaptive randomization. For a large class of covariate-adaptive randomization methods, we obtain the asymptotic distribution of the test statistic under the null hypothesis and derive the conditions under which the test is conservative, valid, or anti-conservative. Several commonly used generalized linear models, such as logistic regression and Poisson regression, are discussed in detail. An adjustment method is also proposed to achieve a valid size based on the asymptotic results. Numerical studies confirm the theoretical findings and demonstrate the effectiveness of the proposed adjustment method.
\end{abstract}


\maketitle

\section{Introduction}\label{sec:Introduction}
Covariate-adaptive randomization is used in clinical trials to balance 
treatment arms with respect to key covariates. 
Many covariate-adaptive randomization methods have been proposed.
Stratified randomization aims to reduce 
the imbalance of treatment assignments within strata formed by covariates.
A restricted randomization method, such as permuted block design 
\citep{Zelen1974} or Efron's biased coin design \citep{Efron1971}, is 
implemented within each stratum.
Pocock and Simon's minimization \citep{Taves1974, Pocock1975} was proposed to reduce the imbalance at a marginal level of covariates.
This approach was extended by Hu and Hu \cite{Hu2012} to simultaneously reduce various levels of imbalances.
As covariate-adaptive randomization is usually based on discrete covariates 
\citep{Taves2010, Lin2015, Ciolino2019}, this paper focuses mainly on randomization 
methods that balance discrete covariates.
Some covariate-adaptive randomization methods have also been 
proposed to balance continuous covariates \citep{Frane1998, 
	Stigsby2010, Su2011, Ma2013, Zhao2015}.
For more discussion of handling covariates in clinical
trials, see McEntegart \cite{McEntegart2003}, Rosenberger and Sverdlov \cite{Rosenberger2008}, Hu et al. \cite{Hu2014} and the
references therein.

Although covariate-adaptive randomization is widely used in clinical trials, concerns have been expressed over the validity of statistical inference after 
these randomization procedures.
In practice, conventional methods, such as two sample \textit{t}-test or 
generalized linear models, are commonly used without consideration of the covariate-adaptive randomization scheme.
Regulatory guidelines \citep{ICH1998,EMA2015} recommend that all the 
covariates used in covariate-adaptive randomization be adjusted in the 
working model.
In fact, it was theoretically demonstrated that such tests are valid 
provided that the working model is correctly specified and that all of the covariates used in the covariate-adaptive randomization are included \cite{Shao2010}.


{\color{black} However, unadjusted tests are prevalent in practice \citep{Sverdlov2015}. It was reported that only 24\% to 34\% of randomized clinical trials adjusted for covariates in their main analyses, and that, even for the trials that used at least one stratification factor in the randomization, only 26\% of these adjusted for all covariates in the primary analysis \citep{Kahan2014}.}
Several practical reasons are given not to use the full model.
It is difficult to incorporate some covariates into the working model; for example, investigation sites are usually omitted from the analysis for a multicenter clinical trial.
The simplicity of a testing procedure is another reason to include fewer covariates.
Also, the adjustment of too many covariates usually requires a complicated model that lacks robustness to model misspecification.
Further discussion can be found in Shao et al. \cite{Shao2010}, Shao and Yu \cite{Shao2013}, and Ma et al. \cite{ Ma2015}.
It is therefore desirable to study the inference properties of conventional tests when some or all of the covariates used in 
covariate-adaptive randomization are excluded from the working model.

The validity of statistical inference under covariate-adaptive 
randomization is well studied in the linear model framework.
Most notably, the unadjusted \textit{t}-test has been theoretically demonstrated to be conservative 
for a variety of covariate-adaptive designs, including the stratified permuted block
design and Pocock and Simon's minimization method \citep{Shao2010, Ma2015, Bugni2018, Ma2019, Ma2020Regression}.
{\color{black} Extensions to multiple treatments have also been studied \cite{Bugni2019, Ma2020Sinica}. Moreover, unadjusted testing was also found to be conservative in longitudinal analysis \cite{Weng2017}, survival analysis \cite{Ye2020JRSSB}, and quantile regression \cite{Zhang2020}.  }

The theoretical foundation for inference under covariate-adaptive randomization 
is less comprehensive in the generalized linear model framework.
Shao and Yu \cite{Shao2013} showed that a model-free $t$-test with no covariate 
is conservative under a specific type of stratified randomization.
However, the inference properties are not studied for testing based on 
generalized linear models, which are more popular in practice.
Moreover, their results do not apply to other commonly used 
covariate-adaptive randomization methods, such as Pocock and Simon's 
minimization, or some newly proposed designs, such as that of Hu and Hu \cite{Hu2012}.
Gail \cite{Gail1988} studied model-based unadjusted tests for treatment effect under 
perfectly balanced studies, but such perfect 
balance is not guaranteed under covariate-adaptive randomization; thus, the 
application of this study's conclusions is also limited.
{\color{black} Fan et al. \cite{Fan2018} considered the use of generalized linear models under covariate-adaptive randomization when covariates are subject to misclassification. However, their results are mainly based on simulations and lack theoretical justifications.}

The purpose of this study is to examine the theoretical properties of the unadjusted test for treatment effect under generalized linear models for a large family of 
covariate-adaptive randomization.
We establish the asymptotic properties of such tests for general 
covariate-adaptive randomization, including both stratified randomization and 
Pocock and Simon's minimization method, among other commonly used methods.
We derive the conditions under which the test is conservative, valid, or 
anti-conservative and propose an adjustment method to achieve a valid size. 
Several important generalized linear models, such as logistic regression and 
Poisson regression, are discussed in detail.
In addition, simulations show that the adjusted tests under covariate-adaptive 
randomization are more powerful than those under complete randomization.
The results provide insights and guidance for testing the treatment effect when using generalized linear models under covariate-adaptive randomization.

This paper is organized as follows.
In Section \ref{sec:TestingGLM}, we first describe the framework for studying the unadjusted
Wald test for treatment effect under covariate-adaptive 
randomization, and then give the asymptotic properties of the test when the canonical 
link is used in the generalized linear model.
In Section \ref{sec:TestingCommonModels}, we apply the derived results to 
several generalized linear models commonly used in clinical trials and 
discuss their testing properties.
In Section \ref{sec:TestingGLM-NC}, we establish general asymptotic results by 
extending the results in Section \ref{sec:TestingGLM} to cases of 
non-canonical link.
An adjustment method to achieve valid testing is proposed in Section 
\ref{sec:Adjustment}.
Section \ref{sec:Simulation} presents simulation results to examine 
the finite sample performance of the proposed theory.
The last section concludes with some remarks and directions for future work.
The proofs and additional simulation results are given in Appendix.

\section{Wald Test under Generalized Linear Models with 
	Canonical Link}\label{sec:TestingGLM}
\subsection{Framework and Notations}
We consider a clinical trial with two treatment groups, 1 and 2.
A covariate-adaptive randomization procedure is implemented to allocate 
patients based on their covariate profiles.
Let $X$ be the vector of the covariates and $T$ be the treatment assignment 
indicator, where $T=1$ if a patient is assigned to treatment 1, and $T=0$ 
otherwise.
Suppose that $Y$ is the response variable and the conditional expectation of $Y$ 
given $X$ and $T$ is
\begin{align}\label{eq:true_model}
\Ex[Y|T,X]=h(\mu+\delta T+\beta X)=h(\eta),
\end{align}
where $\delta$ is the treatment effect, $\beta$ is the vector of 
unknown parameters for the covariates, and $\eta=\mu+\delta 
T+\beta X$ is the linear predictor. 
The covariates are assumed to be independent and identically distributed for 
each patient.
Without loss of generality, it is also assumed that $\Ex[X]=0$.

Suppose that conditional on $\eta$, the distribution of $Y$ belongs to an 
exponential family and the likelihood of $Y$ is given by
\begin{align}
\exp\left\{\frac{Y\theta-b(\theta)}{\phi}+c(Y,\phi)\right\},
\end{align}
where $\theta$ depends on $\eta$, and $\phi$ is the dispersion parameter, which is assumed to be known.
Thus,  the relationship between $Y$ and $(T, X)$ can be modeled by 
generalized linear models \citep{McCullagh1989}.

We first consider the scenario where the canonical link function is 
used, i.e., $\theta=\eta$.
This covers several common generalized linear models, including 
linear regression with $h(\eta)=\eta$, logistic regression with 
$h(\eta)=\exp(\eta)[1+\exp(\eta)]^{-1}$, and Poisson regression with 
$h(\eta)=\exp(\eta)$.
The cases with non-canonical link functions are studied in Section 
\ref{sec:TestingGLM-NC}.

\subsection{Working Model and Wald Test}

Suppose that the covariates used in covariate-adaptive randomization are omitted 
from post-randomization data analysis; then the following working model 
is used:
\begin{align}\label{eq:working_model}
\Ex[Y|T]=h(\mu+\delta T)=h(\eta),
\end{align}
where the link function and conditional distribution of $Y$ given $\eta$ 
are the same as in the true model \eqref{eq:true_model}.

\begin{remark}
	It is important to note that model \eqref{eq:working_model} is a mis-specified working model without adjustment of any covariates used in randomization.
	The properties of the test for treatment effect under such a model mis-specification have been extensively studied for linear models both numerically and theoretically in the literature, including Birkett \cite{Birkett1985}, Forsythe \cite{Forsythe1987}, Shao et al. \cite{Shao2010}, and Ma et al. \cite{Ma2015}.
	As for generalized linear models, Gail \cite{Gail1988} studied the score test under model \eqref{eq:working_model} in perfectly balanced trials.
	However, the results are not applicable to general covariate-adaptive randomization, because perfect balance is usually not achieved. 
\end{remark}

To test the treatment effect based on the working model 
\eqref{eq:working_model}, we consider the Wald test, because it is commonly used in practice and is easily accessible in most statistical packages.
The Wald test can be formed as follows:
\begin{align}\label{eq:test}
H_0: \delta=0 \text{ versus } H_1: \delta \neq 0,
\end{align}
with the test statistic
\begin{align}\label{eq:test_stat}
S=\frac{\hat{\delta}}{\widehat{\se}(\hat\delta)},
\end{align}
where $\hat\delta$ is the maximum likelihood estimate (MLE) of $\delta$, 
and $\widehat{\se}(\hat\delta)$ is the model-based estimate of the standard 
error of $\hat{\delta}$.
To perform a two-sided test, the null hypothesis is rejected if 
$|S|>Z_{1-\alpha/2}$, where $Z_{1-\alpha/2}$ is the $({1-\alpha/2})$th quantile of 
a standard normal distribution.

Suppose that $n$ patients  in total are enrolled in the covariate-adaptive randomized 
trial.
Let $(y_i, T_i, x_i)$ be the observed value of $(Y, T, X)$ for the $i$-th 
patient.
The working model then has the following log-likelihood
\begin{align}\label{eq:log-likelihood}
\ell(\mu, \delta)=\sum_{i=1}^{n}\left\{
\frac{y_i\theta_i-b(\theta_i)}{\phi}+c(y_i,\phi)\right\},
\end{align}
where $\theta_i=\eta_i=\mu+\delta T_i$ under the canonical link.
Also, it follows from the properties of the exponential family that 
$b^{'}(\theta_i)=h(\eta_i)$.
By setting the first derivative of the log-likelihood, with respect to 
${\delta}$ and ${\mu}$, to be zero, the MLE estimates of $\mu$ and $\delta$ 
are  
\begin{align}\label{eq:muhat}
\hat{\mu}=h^{-1}\left[\frac{\sum_{i=1}^{n} (1-T_i)y_i}{n_0}\right],
\end{align}
and 
\begin{align}\label{eq:deltahat}
\hat{\delta}=h^{-1}\left[\frac{\sum_{i=1}^{n} T_iy_i}{n_1}\right]-
h^{-1}\left[\frac{\sum_{i=1}^{n} (1-T_i)y_i}{n_0}\right],
\end{align}
where $n_1=\sum_{i=1}^{n} T_i$, and $n_0=n-n_1$.

\subsection{Main Results}\label{sec:TestingGLM-results}
%
Before giving the main results, we describe how 
to measure various levels of covariate balance under a covariate-adaptive 
design, which is shown to be closely related to the test properties. 
{\color{black} Note that in this paper, we assume that the covariates $X$ are discrete variables.}

Consider $p$ covariates included in $X$ and $m_k$ levels for the $k$th 
covariate, 
resulting in $m=\prod_{i=1}^{p}m_k$ strata in total.
Denote the covariate profile of the $i$th patient $(t_1, \ldots, t_p)$ if 
the $k$th covariate is at level $t_k$, $1\le k\le p$ and $1\le t_k \le m_k$.
For convenience, we use $(t_1, \ldots, t_p)$ to denote the stratum formed by 
patients who possess the same covariate profile $(t_1, \ldots, t_p)$, use 
$(k; t_k)$ to denote the margin formed by patients whose $k$th covariate is at 
level $t_k$.
Then let
\begin{description}
	\item[-] $D_n$ be the difference between the numbers of patients in 
	treatment groups 1 and 2 as total, i.e., the number in group 1 minus the 
	number in group 2;
	\item[-] $D_n(k;t_k)$ be the differences between the 
	numbers of patients in the two treatment groups on the margin $(k;t_k)$;
	\item[-] $D_n(t_1, t_2,\ldots,t_p)$ be the difference 
	between the numbers of patients in the two treatment groups within the 
	stratum $(t_1, t_2,\ldots,t_p)$.
\end{description}

\begin{remark}
	Stratified randomization usually has strong balance properties.
	For example, stratified permuted block design and stratified biased 
	coin design, which use the permuted block design and Efron's biased coin 
	design within each stratum, have within-stratum imbalances that are 
	bounded in probability.
	As marginal imbalance and overall imbalance can be considered as sums of 
	certain within-stratum imbalances, the marginal imbalance and overall 
	imbalance are also bounded in probability for these two designs.
	In addition, Pocock and Simon's minimization can achieve 
	marginal and overall imbalances that are bounded in probability 
	\citep{Ma2015}. 
	In addition, Hu and Zhang \cite{Hu2020} proved that the 
	within-stratum imbalances increase with the order of $\sqrt{n}$.
\end{remark}	

Now we first present the theoretical results for the test for treatment 
effect if within-stratum imbalances are bounded in probability under a 
covariate-adaptive design. 

\begin{theorem}\label{thm:TestWithinStratumBounded}
	Suppose that a covariate-adaptive design satisfies the condition that all 
	within-stratum imbalances are bounded in probability; then under the null 
	hypothesis $H_0: \delta=0$, we have, as $n \rightarrow \infty$,
	\begin{align}
	S \Dcon N \left(0, \frac{\Ex[\var(Y|X)]}{\phi h^{'}h^{-1}(\Ex[Y])}\right).
	\end{align}
\end{theorem}

Based on Theorem \ref{thm:TestWithinStratumBounded}, if the within-stratum 
imbalances are bounded in probability, the size of the test for treatment effect depends on a comparison between $\Ex[\var(Y|X)]$ and $\phi 
h^{'}h^{-1}(\Ex[Y])$, i.e., in the asymptotic sense,
\begin{enumerate}
	\item[(i)] the test is conservative if 
	$\Ex[\var(Y|X)]<\phi h^{'}h^{-1}(\Ex[Y])$,
	\item[(ii)]the test is valid if 
	$\Ex[\var(Y|X)]=\phi h^{'}h^{-1}(\Ex[Y])$,
	\item[(iii)]the test is anti-conservative if 
	$\Ex[\var(Y|X)]>\phi h^{'}h^{-1}(\Ex[Y])$.
\end{enumerate}
Several commonly used generalized linear models are discussed in detail in Section 3.

\begin{remark}
	The condition that within-stratum imbalances are bounded in probability is satisfied by many covariate-adaptive randomization methods based on 
	stratification, such as stratified permuted block design.
	However, it is important to note that the results for Theorem \ref{thm:TestWithinStratumBounded} have applications beyond stratification.
	In particular, Hu and Hu \cite{Hu2012} proposed a new family of covariate-adaptive methods that can simultaneously reduce imbalances of all levels (overall, marginal, and within-stratum).
	The within-stratum imbalances are bounded in probability, so the results given in Theorem \ref{thm:TestWithinStratumBounded} can also be applied to their designs.
\end{remark}

Although the condition in Theorem \ref{thm:TestWithinStratumBounded} is 
satisfied by many stratified randomization methods, it does not hold for some other covariate-adaptive designs, such as Pocock and Simon's minimization.
We next study the test for treatment effect under a broader class of 
covariate-adaptive designs. 
For this, we first state two conditions below.
\bigskip

\noindent\textbf{\underline{Condition} (A)} The overall imbalance is bounded in probability, that is, 
$D_n=O_P(1)$.


\noindent\textbf{\underline{Condition} (B)} Under the null hypothesis, $H_0: \delta=0$, as $n 
\rightarrow \infty$,
\begin{align}
\frac{\sum_{i=1}^{n}{\left(2T_i-1\right)h(\mu+\beta x_i)}}{\sqrt{n}}
\Dcon N(0, \sigma_h^2).
\end{align}

These conditions are mild and satisfied by many covariate-adaptive 
randomization methods that balance discrete covariates.
For covariate-adaptive designs with within-stratum imbalances that are bounded 
in probability, the distribution in Condition (B) reduces to $O_P(1/\sqrt{n})$, 
and hence $\sigma_h^2=0$, by 
noting that  $h(\mu+\beta x_i)$ only takes a finite number of values that 
correspond to each stratum.
Furthermore, the conditions hold for Pocock and Simon's minimization 
according to Hu and Zhang \cite{Hu2020}, who proved that within-stratum 
imbalances, scaled by $n^{-1/2}$, are asymptotically jointly normal, with positive variances.

We now present the general results of the test for treatment effect if these 
conditions are satisfied under a covariate-adaptive design.

\begin{theorem}\label{thm:TestGeneral}
	Suppose that a covariate-adaptive design satisfies Conditions (A)-(B); then 
	under the null 
	hypothesis $H_0: \delta=0$, we have, as $n \rightarrow \infty$,
	\begin{align}
	S \Dcon N \left(0, \frac{\Ex[\var(Y|X)]+\sigma_h^2}{\phi 
		h^{'}h^{-1}(\Ex[Y])}\right).
	\end{align}
\end{theorem}

It is clear to see that the result of Theorem 
\ref{thm:TestWithinStratumBounded} is a special case of Theorem 
\ref{thm:TestGeneral} with $\sigma_h^2=0$.
Compared with stratified randomization, an extra layer of 
variation induced by within-stratum imbalances must be considered
when studying the size of the test for treatment effect under general 
covariate-adaptive designs, such as Pocock and Simon's minimization.
As in the aforementioned discussion, the numerator and denominator of the 
asymptotic variance in Theorem \ref{thm:TestGeneral} must be compared to 
determine whether the test is asymptotically conservative, valid, or 
anti-conservative. 

\begin{remark}
	Under Pocock and Simon's minimization, the within-stratum imbalances increase at the rate of $O_P(\sqrt{n})$, leading to $\sigma_h^2>0$, so the asymptotic variance of test statistic \eqref{eq:test_stat} under minimization is larger than that under stratified randomization, which results in a larger size compared with stratified randomization.
\end{remark} 

\begin{remark}
	Condition (A) is not satisfied by complete randomization.
	However, the independence of treatment assignment and covariates shows that the test statistic $S$  also has a normal distribution under the null hypothesis, 
	\begin{align}\label{eq:cr}
	S \Dcon N \left(0, \frac{\var[Y]}{\phi 
		h^{'}h^{-1}(\Ex[Y])}\right).
	\end{align}
	As $\var[Y]=\Ex[\var(Y|X)]+\var[\Ex(Y|X)]$, the asymptotic variance of $S$ under complete randomization is generally greater than that under stratified randomization.
\end{remark}

\section{Test Size under Some Commonly Used 
	Models}\label{sec:TestingCommonModels}

In this section, we study the impact of using working model 
\eqref{eq:working_model} on the size of the test for treatment effect based on 
Theorems \ref{thm:TestWithinStratumBounded} and \ref{thm:TestGeneral}.
Several commonly used models with canonical link functions are covered.

\subsection{Logistic Regression}\label{sec:TestingCommonModels-logistic}
Logistic regression is extensively used in clinical trials with binary 
responses.
Under logistic regression, the conditional distribution of $Y$ given $\eta$ is 
a Bernoulli distribution ($\phi=1$), and the canonical link takes the form of 
$h(\eta)=\exp(\eta)[1+\exp(\eta)]^{-1}$.
Also, it is easy to verify that
\begin{align*}
\phi h'h^{-1}(\Ex[Y])=\var(Y).
\end{align*}
Therefore, provided that the condition in Theorem 
\ref{thm:TestWithinStratumBounded} is satisfied, i.e., the within-stratum 
imbalances are bounded in probability, under $H_0: \delta=0$,
\begin{align*}
S \Dcon N \left(0, \frac{\Ex[\var(Y|X)]}{\var(Y)}\right).
\end{align*}
Because $\Ex[\var(Y|X)]$ is smaller than ${\var(Y)}$ in general, the asymptotic 
distribution of $S$ shows a greater concentration around 0 than the standard normal 
distribution, resulting in a conservative test whose test size is smaller than the 
nominal level. 
The conclusion holds for the stratified permuted block design and the class of 
covariate-adaptive designs proposed by Hu and Hu \cite{Hu2012}.

The size under minimization is larger than that under stratified randomization 
because of a positive $\sigma_h^2$, and the simulation studies in Section 
\ref{sec:Simulation} show that the test under minimization is also 
conservative. 
Under complete randomization, the variance of the asymptotic distribution of 
$S$ is equal to 1 according to \eqref{eq:cr}, leading to a valid test.

\subsection{Poisson Regression}\label{sec:TestingCommonModels-poisson}
Poisson regression can be used to model event counts whose conditional 
distribution given $\eta$ is a Poisson distribution ($\phi=1$).
Under the canonical link $h(\eta)=\exp(\eta)$ and $H_0: \delta=0$, we have
\begin{align*}
\phi h'h^{-1}(\Ex[Y])=\Ex[Y]=\Ex[\var(Y|X)],
\end{align*}
as $\Ex[Y|X]=\var[Y|X]$. 
By Theorem \ref{thm:TestWithinStratumBounded} it is easy to see that the test for 
treatment effect is valid if within-stratum 
imbalances are bounded in probability.
In contrast, both minimization and complete randomization have inflated Type I 
errors, resulting in anti-conservative tests.

\subsection{Linear Regression}
Under linear regression, a response is a continuous variable whose 
conditional distribution given $\eta$ is a normal distribution with mean $\eta$ 
and variance $\phi=\sigma^2$.
The canonical link is the identity function $h(\eta)=\eta$.
Assuming that $\phi=\sigma^2$ is known, we have,
\begin{align*}
\phi h'h^{-1}(\Ex[Y])=\sigma^2=\Ex[\var(Y|X)],
\end{align*}
so the test is valid under stratified randomization.
However, the test is anti-conservative for minimization by Theorem 
\ref{thm:TestGeneral}.

In practice,  $\phi=\sigma^2$ is usually unknown and must be estimated based 
on the data.
Such hypothesis testing in the linear model framework has been theoretically 
studied in the literature \citep{Shao2010, Ma2015}.
It is proved that if influential randomization covariates are omitted from the 
working model, the test for treatment effect is conservative for a large class 
of covariate-adaptive designs, including both stratified randomization and 
minimization.

\subsection{Exponential Model}\label{sec:TestingCommonModels-exponential}
In the exponential model, the response variable $Y$ given $\eta$ is 
exponentially distributed ($\phi=1$).
Under the canonical link $h(\eta)=-1/\eta$ and $H_0: \delta=0$, we have
\begin{align*}
\phi h'h^{-1}(\Ex[Y])=(\Ex[Y])^2=(\Ex\{\Ex[Y|X]\})^2.
\end{align*}
Notice that $\Ex[\var(Y|X)]=\Ex\{[\Ex(Y|X)]^2\}$; then by the Cauchy-Schwarz 
inequality, 
\begin{align*}
\phi h'h^{-1}(\Ex[Y])\le\Ex[\var(Y|X)],
\end{align*}
so the test for treatment effect is typically anti-conservative for stratified 
randomization with inflated Type I errors by Theorem 
\ref{thm:TestWithinStratumBounded}.
Furthermore, the extent of inflated Type I errors is even more severe under 
minimization and complete randomization by Theorems \ref{thm:TestGeneral} and 
\eqref{eq:cr}.

\section{General Results with Non-Canonical Link}
\label{sec:TestingGLM-NC}

In this section, the results in Section 
\ref{sec:TestingGLM} obtained under canonical link functions are extended to cases of non-canonical link functions.

We first note that the MLE estimates of $\mu$ and 
$\delta$ based on the working model \eqref{eq:working_model} with a 
non-canonical link take the same form as those with a canonical link.
Let $\theta=\gamma(\eta)$, which is not necessarily the identity function under 
general link functions.
Gail et al. \cite{Gail1984} showed that if $\gamma'(\hat{\mu})$ and 
$\gamma'(\hat{\mu}+\hat{\delta})$ are not degenerate, 
the MLE estimates $\hat{\mu}$ and $\hat{\delta}$ under a non-canonical link 
are the same as \eqref{eq:muhat} and \eqref{eq:deltahat}.

\begin{remark}
	In fact, the estimates $\hat{\mu}$ and $\hat{\delta}$ under a canonical link can 
	be considered as the estimates by the method of moments.
	It was shown that the MLE estimates under the working model 
	\eqref{eq:working_model} are equal to the moment estimates if the 
	non-degenerate conditions are satisfied \citep{Gail1984} .
\end{remark} 

Next, we state two general theorems on the test for treatment effect when a
non-canonical link function is used. These theorems correspond to Theorems 
\ref{thm:TestWithinStratumBounded} and \ref{thm:TestGeneral}, respectively.

\begin{theorem}\label{thm:TestWithinStratumBounded-NC}
	Suppose that a covariate-adaptive design satisfies the condition that all 
	within-stratum imbalances are bounded in probability; then under the null 
	hypothesis $H_0: \delta=0$, we have, as $n \rightarrow \infty$,
	\begin{align}
	S \Dcon N \left(0, \frac{\Ex[\var(Y|X)]}{\phi 
		h^{'}h^{-1}(\Ex[Y])/\gamma^{'}h^{-1}(\Ex[Y])}\right).
	\end{align}
\end{theorem}

\begin{theorem}\label{thm:TestGeneral-NC}
	Suppose that a covariate-adaptive design satisfies Conditions (A)-(B) {\color{black} listed above Theorem \ref{thm:TestGeneral} in Section \ref{sec:TestingGLM-results}}; then 
	under the null 
	hypothesis $H_0: \delta=0$, we have, as $n \rightarrow \infty$,
	\begin{align}
	S \Dcon N \left(0, \frac{\Ex[\var(Y|X)]+\sigma_h^2}{\phi 
		h^{'}h^{-1}(\Ex[Y])/\gamma^{'}h^{-1}(\Ex[Y])}\right).
	\end{align}
\end{theorem}

It is clear that the two theorems above reduce to Theorems 
\ref{thm:TestWithinStratumBounded} and \ref{thm:TestGeneral} by setting 
$\gamma(\eta)=\eta$.
Based on the fact that $\Ex[\var(Y|X)]=\Ex[\phi h^{'}(\mu+\beta 
X)/\gamma^{'}(\mu+\beta 
X)]$ and $\Ex[Y]=\Ex[h(\mu+\beta X)]$, one can use these general theorems 
to evaluate the test size for various generalized linear models under 
covariate-adaptive randomization with an argument similar to that in Section 
\ref{sec:TestingCommonModels}.

\section{Adjusted Test}
\label{sec:Adjustment}

As discussed in Sections \ref{sec:TestingGLM} to \ref{sec:TestingGLM-NC}, the 
Wald test under the working model \eqref{eq:working_model} is invalid, either 
conservative or anti-conservative, for many generalized linear models.
In this section, we show how to make an adjustment to achieve a valid test 
based on the theoretical results obtained.
{\color{black} The purpose of the proposed adjusted test is twofold. First, it provides a remedy approach when using generalized linear models that do not adjust for the randomization covariates. However, we note that a simpler and more practical alternative is to adjust for the covariates within the analysis. Second, the adjusted test's validity, later confirmed by simulation studies, provides an added layer of justification for the theoretical results.}

We first consider the scenario in which within-stratum imbalances are bounded in 
probability, such as stratified permuted block design and the class of 
covariate-adaptive designs proposed by Hu and Hu \cite{Hu2012}.
For simplicity of notation, let $\sigma_\nu^2=\Ex[\var(Y|X)]$.
{\color{black} We note that in the proof of Theorem \ref{thm:TestGeneral-NC} in the Appendix, it is shown that, }
under model \eqref{eq:working_model} and 
$H_0: \delta=0$, we have
\begin{align}\label{eq:dist_adj_str}
\frac{h^{'}h^{-1}(\Ex[Y])\hat{\delta}}{2\sqrt{\sigma_\nu^2/n}}
\Dcon N(0,1).
\end{align}
An adjusted test statistic can then be constructed by replacing the 
population parameters by their sample estimates,
\begin{align}\label{eq:test_stat_adj_str}
S_{adj}=\frac{h^{'}h^{-1}(\bar Y_n)\hat{\delta}}{2\sqrt{\hat\sigma_\nu^2/n}},
\end{align}
where $\bar Y_n=\sum_{i=1}^{n}y_i/n$ is the overall sample mean. 
$\hat\sigma_\nu^2=\sum_{i=1}^{n}\sum_{j=1}^{m} 
(y_i-\hat\mu_j)^2I\{x_i=s_j\}/(n-m)$ is the sample estimate of $\sigma_\nu^2$,
where $s_j$ represents the $j$th stratum formed by covariates, and 
$\hat\mu_j$ is the sample mean of $Y_i$ in the $j$th stratum.
Thus $S_{adj}$ follows a standard normal distribution asymptotically under the 
null hypothesis, leading to a valid test.

Furthermore, for a general covariate-adaptive randomization that satisfies 
Conditions (A) and (B), a conclusion similar to \eqref{eq:dist_adj_str} can 
also be established based on the proof {\color{black} of Theorem \ref{thm:TestGeneral-NC}} in the Appendix, i.e., under model 
\eqref{eq:working_model} and $H_0: \delta=0$,
\begin{align}\label{eq:dist_adj_general}
\frac{h^{'}h^{-1}(\Ex[Y])\hat{\delta}}{2\sqrt{(\sigma_\nu^2+\sigma_h^2)/n}}
\Dcon N(0,1).
\end{align}
Similarly to \eqref{eq:test_stat_adj_str}, a valid test can be constructed with 
an adjusted test statistic
\begin{align}\label{eq:test_stat_adj_general}
S_{adj}^{*}=\frac{h^{'}h^{-1}(\bar 
	Y_n)\hat{\delta}}{2\sqrt{(\hat\sigma_\nu^2+\hat\sigma_h^2)/n}},
\end{align}
as long as $\hat\sigma_h^2$ is a consistent estimator of $\sigma_h^2$. 
Compared with \eqref{eq:test_stat_adj_str}, the only additional quantity we need 
to estimate is $\sigma_h^2$.
Note that $
\sum_{i=1}^{n}{\left(2T_i-1\right)h(\mu+\beta x_i)}
=\sum_{j=1}^{m}D_n(s_j)h(\mu+\beta s_j)
$
is a weighted sum of within-stratum imbalances $D_n(s_j)$.
If the joint distribution of within-stratum imbalances of 
$\{D_n(s_j),j=1,\ldots,m\}$ is known, we can estimate $\sigma_h^2$ by 
using the joint distribution and replacing $h(\mu+\beta s_j)$ with its sample 
estimate $\hat\mu_j$.
Another option is to evaluate the joint within-stratum imbalances of 
$\{D_n(s_j),j=1,\ldots,m\}$ numerically. We provide further details on the minimization in Section \ref{sec:Simulation-Size-Adj}.

\begin{remark}
	A valid test can also be achieved for complete randomization, with a 
	similar test statistic defined as in \eqref{eq:test_stat_adj_str} and 
	\eqref{eq:test_stat_adj_general},
	\begin{align}\label{eq:dist_adj_cr}
	S_{adj, CR}=\frac{h^{'}h^{-1}(\bar 
		Y_n)\hat{\delta}}{2\sqrt{\hat\sigma_Y^2/n}},
	\end{align}
	where $\hat\sigma_Y^2$ is the overall sample variance.
\end{remark}

After adjustment, we can obtain more powerful hypothesis testing results.
A simulation study is carried out in the next section to evaluate the 
size and power of these adjusted tests under various randomization methods 
and generalized linear models.

\section{Numerical Studies}\label{sec:Simulation}

\subsection{Size of Wald Test}\label{sec:Simulation-Size}

We consider three commonly used generalized linear models: logistic 
regression, Poisson regression, and exponential model.
The asymptotic properties of the Wald test for treatment effect under model 
\eqref{eq:working_model} are discussed in Section \ref{sec:TestingGLM}.
We now perform simulations to evaluate the size under a finite sample.

Assume that the true model is 
\begin{align}\label{eq:20}
\Ex(y_i|T_i,x_i)=h(\mu+\delta T_i+\beta_1 x_{i,1}+\beta_2 x_{i,2})=h(\eta_i),
\end{align}
where $\delta$ represents the treatment effect and 
$x_{i,j}$ are independent and identically distributed random variables that 
follow the Bernoulli distribution with $P=0.5$.
The parameters $(\mu,\beta_1,\beta_2)=(-1,2,4)$ are assumed for 
logistic regression, and    
$(\mu,\beta_1,\beta_2)=(0.2,0.5,1)$ for both Poisson regression and 
exponential model.
Canonical link functions are used for these models.
Note that the inverse function, instead of the negative inverse function, is used 
for the exponential model in the simulation, so the range of $\Ex(y_i|T_i,x_i)$ 
is permitted. 
{\color{black} The error distributions (i.e., the conditional distributions of $y_i$ given $\eta_i$) for these three generalized linear models are respectively specified in Sections \ref{sec:TestingCommonModels-logistic}, \ref{sec:TestingCommonModels-poisson}, and \ref{sec:TestingCommonModels-exponential}.

For each of the three generalized linear models under consideration, we simulate datasets of two sample sizes, $N=200$ and $500$, and four different randomization methods, which are complete randomization, stratified permuted block design, Pocock and Simon's minimization, and the method proposed by Hu and Hu. In each scenario, we first generate the patients’ covariates $x_{i,1}$ and $ x_{i,2}$, $i=1,\ldots,N$, according to the Bernoulli distributions as specified above. Then we use one of the randomization methods to obtain the treatment assignments $T_i$. In simulations, a block size of 4 is used in stratified permuted block design, and a biased coin probability of 0.75 is used in Pocock and Simon's minimization and the method proposed by Hu and Hu. Finally, the treatment responses $y_i$ are simulated according to the true model \eqref{eq:20} and the corresponding error distributions with parameters as described above. }
The following working model is used to test $H_0: \delta=0$,
\begin{align}
\Ex(y_i|T_i)=h(\mu +\delta T_i),
\end{align}
where no covariates are included in the model. 
The significance level for the two-sided tests is $\alpha = 0.05$.
{\color{black} Each scenario is simulated 5,000 times.}

\begin{center} [Tables 1-3 here.] \end{center}

The simulated sizes are presented in Tables 1-3 under the column ``Wald 
Test'' and are consistent with the 
theoretical results.
In particular, the test for treatment effect based on logistic regression is 
generally conservative under covariate-adaptive randomization if covariates are 
omitted from the working model, whereas the test is still valid under complete 
randomization.
For Poisson regression, the omission of covariates leads to a valid test under the
stratified permuted block design and under Hu and Hu's method.
However, the inflation of Type I error is slight under minimization and more 
severe under complete randomization.
The tests are all anti-conservative for the exponential model if no covariates 
are included. 
In all scenarios, the stratified permuted block design and Hu and Hu's method perform comparably in terms of test size, whereas the test size for 
minimization lies between those of these two designs and complete randomization.
{\color{black} Finally, we note that the simulation results under complete randomization are expected (see Section \ref{sec:TestingCommonModels}) and concur with the findings in the literature \cite{Gail1988b, Fan2018}.}

\subsection{Size of Adjusted Test}
\label{sec:Simulation-Size-Adj}

In this section, adjustment is performed to achieve a valid test for treatment 
effect, as discussed in Section \ref{sec:Adjustment}.
Using the same setup as in the previous section, the adjustment follows the 
procedure in Section \ref{sec:Adjustment}.
For minimization, Monte Carlo simulation is used to generate $B(=500)$ samples 
of joint within-stratum imbalances of $\{D_n(s_j),j=1,\ldots,m\}$ via the 
repetitive application of minimization on the covariates; then 
$\hat\sigma_h^2$ can be obtained by calculating the variance of 
$\sum_{j=1}^{m}D_n(s_j)\hat\mu_j$ and can be used to construct the test 
statistic \eqref{eq:test_stat_adj_general}.
The sizes for the adjusted tests are presented in Tables 1-3 under the column 
``Adjusted Test.''
Based on the results, the adjustment works quite well because the 
sizes are successfully controlled at 5\% for all three generalized linear 
models under different randomization methods.

\subsection{Power of Adjusted Test}
\label{sec:Simulation-Power-Adj}

Based on the adjusted tests, we also perform simulations to evaluate the power 
of these tests under various generalized linear models and randomization 
methods.
The same setup is used as in the previous sections, except that $\delta$ is 
assumed to be a sequence of non-zero numbers to assess power.
In particular, $\delta$ is assumed to be $0, 0.2, 0.4, \ldots, 2$ for logistic 
regression and $0, 0.04, 0.08, \ldots, 0.4$ for Poisson regression and the
exponential model, so that power ranges from 5\% to more than 
90\%. The powers are presented in Figure 1. 
It can be seen that the adjusted test has the highest power under the
stratified permuted block design and Hu and Hu's method for all three 
generalized linear models, which is followed by minimization.
The adjusted test under complete randomization is the least powerful.
Like size, the power of minimization lies between the
stratified permuted block design and complete randomization.

\begin{center} [Figure 1 here.] \end{center}

\section{Conclusions}\label{sec:Conclusion}

In this study, we examine the test for treatment effect based on 
generalized linear models under a large family of covariate-adaptive 
randomization.
We derive the asymptotic distribution of the unadjusted Wald test statistic under the null 
hypothesis when no covariates are incorporated into the working model,
and present the conditions under which the test is valid or not.
For the invalid cases, we further propose an adjustment to maintain the test 
validity with increased power.

{\color{black} Although we have focused on hypothesis testing, the findings have profound implications for treatment effect estimation, which is often an important target in randomized clinical trials. In particular, our results suggest that the usual Wald confidence intervals may be improperly specified if the covariates used in covariate-adaptive randomization are omitted from generalized linear models. It would be interesting to investigate the estimation issues, such as bias and constructing valid confidence intervals, for generalized linear models under covariate-adaptive randomization.}

It is assumed that discrete covariates are used in the randomization procedure. 
However, it is common to collect continuous covariates in clinical trials. 
{\color{black} Continuous covariates are typically discretized in order to be included in the randomization scheme. Alternatively, the methods that directly balance continuous covariates are useful and are gaining popularity \citep{Frane1998, Stigsby2010, Su2011, Ma2013, Zhao2015}.}
For inference in the linear model framework, discretization was considered in 
Shao et al. \cite{Shao2010} and Ma et al. \cite{Ma2015}, whereas adaptive randomization with continuous covariates was studied in Qin et al. \cite{Ma2017} and Li et al. \cite{Li2019}.
It is desirable to extend the framework proposed in this paper to study 
the test for treatment effect under generalized linear models when continuous 
covariates are considered.

In addition to the treatment effect, it is of interest to test whether a covariate is 
influential on patient outcome.
In this paper, we focus mainly on the test for treatment effect without the use of any covariates in the working model.
The property of testing covariate effects remains unknown if partial covariate information is used in the working model.
It was shown that testing covariates is valid under certain conditions if a 
linear model is used \citep{Ma2015}, but the properties under generalized 
linear models are not clear.
Finally, it might also be possible to consider the scenario in which the covariates are subject to misclassification \citep{Fan2018, Wang2020}.
These topics are left for future research.

{\color{black}
\begin{acks}
We are grateful to the referees for their many helpful comments.
\end{acks}
}

\begin{funding}
The authors disclosed receipt of the following financial support for the research, authorship, and/or publication of this article: This work was supported by the Fundamental Research Funds for the Central Universities, and the Research Funds of Renmin University of China [grant number 20XNA023].
\end{funding}

\begin{dci}
The authors declare that there is no conflict of interest.
\end{dci}

\bibliographystyle{SageV}
\bibliography{WeiBib_SMMR}

\newpage 

\begin{table}[htb]
	\begin{center}
	\caption{{\color{black} Simulated size in \% for logistic regression under various randomization methods (CR, PS, SB, HH), testing methods (Wald test, adjusted test), and sample sizes ($N=200, 500$). }}
		\bigskip
		\begin{tabular}{c c c c }
			$N$   & Randomization      & Wald Test &  Adjusted Test \\
			\hline
			\hline
			200   &   \textbf{CR}      & 4.50 & 5.08\\
			&   \textbf{PS}      & 1.30 & 4.08 \\
			&   \textbf{SB}      & 1.02 & 5.82 \\
			&   \textbf{HH}      & 1.16 & 6.34 \\
			\hline
			500   &   \textbf{CR}      & 5.64 & 5.78\\
			&   \textbf{PS}      & 1.20 & 4.16 \\
			&   \textbf{SB}      & 0.94 & 4.96 \\
			&   \textbf{HH}      & 1.06 & 5.34 \\
			\hline
		\end{tabular}
	\end{center}
	Note: N, sample size; CR, complete randomization; PS, Pocock and Simon's 
	minimization;
	SB, stratified permuted block design; HH, the method proposed by Hu and Hu \cite{Hu2012}.
\end{table}

\bigskip

\begin{table}[htb]
	\begin{center}
	\caption{{\color{black} Simulated size in \% for Poisson regression under various randomization methods (CR, PS, SB, HH), testing methods (Wald test, adjusted test), and sample sizes ($N=200, 500$). }}
		\bigskip
		\begin{tabular}{c c c c }
			$N$   & Randomization      & Wald Test &  Adjusted Test \\
			\hline
			\hline
			200   &   \textbf{CR}      & 15.22 & 5.32\\
			&   \textbf{PS}      & 5.60 & 4.16 \\
			&   \textbf{SB}      & 5.06 & 5.28 \\
			&   \textbf{HH}      & 5.02 & 5.18 \\
			\hline
			500   &   \textbf{CR}      & 15.40 & 5.52\\
			&   \textbf{PS}      & 5.40 & 4.56 \\
			&   \textbf{SB}      & 5.00 & 5.06 \\
			&   \textbf{HH}      & 5.20 & 5.42 \\
			\hline
		\end{tabular}
	\end{center}
	Note: N, sample size; CR, complete randomization; PS, Pocock and Simon's 
	minimization;
	SB, stratified permuted block design; HH, the method proposed by Hu and Hu \cite{Hu2012}.
\end{table}

\begin{table}[htb]
	\begin{center}
	\caption{{\color{black} Simulated size in \% for exponential model under various randomization methods (CR, PS, SB, HH), testing methods (Wald test, adjusted test), and sample sizes ($N=200, 500$). }}
		\bigskip
		\begin{tabular}{c c c c }
			$N$   & Randomization      & Wald Test &  Adjusted Test \\
			\hline
			\hline
			200   &   \textbf{CR}      & 21.02 & 5.68\\
			&   \textbf{PS}      & 17.00 & 6.32 \\
			&   \textbf{SB}      & 14.82 & 5.46 \\
			&   \textbf{HH}      & 13.76 & 5.64 \\
			\hline
			500   &   \textbf{CR}      & 22.22 & 4.94\\
			&   \textbf{PS}      & 16.30 & 5.24 \\
			&   \textbf{SB}      & 13.46 & 4.96 \\
			&   \textbf{HH}      & 14.46 & 4.98 \\
			\hline
		\end{tabular}
	\end{center}
	Note: N, sample size; CR, complete randomization; PS, Pocock and Simon's 
	minimization;
	SB, stratified permuted block design; HH, the method proposed by Hu and Hu \cite{Hu2012}.
\end{table}

\bigskip

\begin{figure}
\centering
\includegraphics[width=15cm,height=6cm]{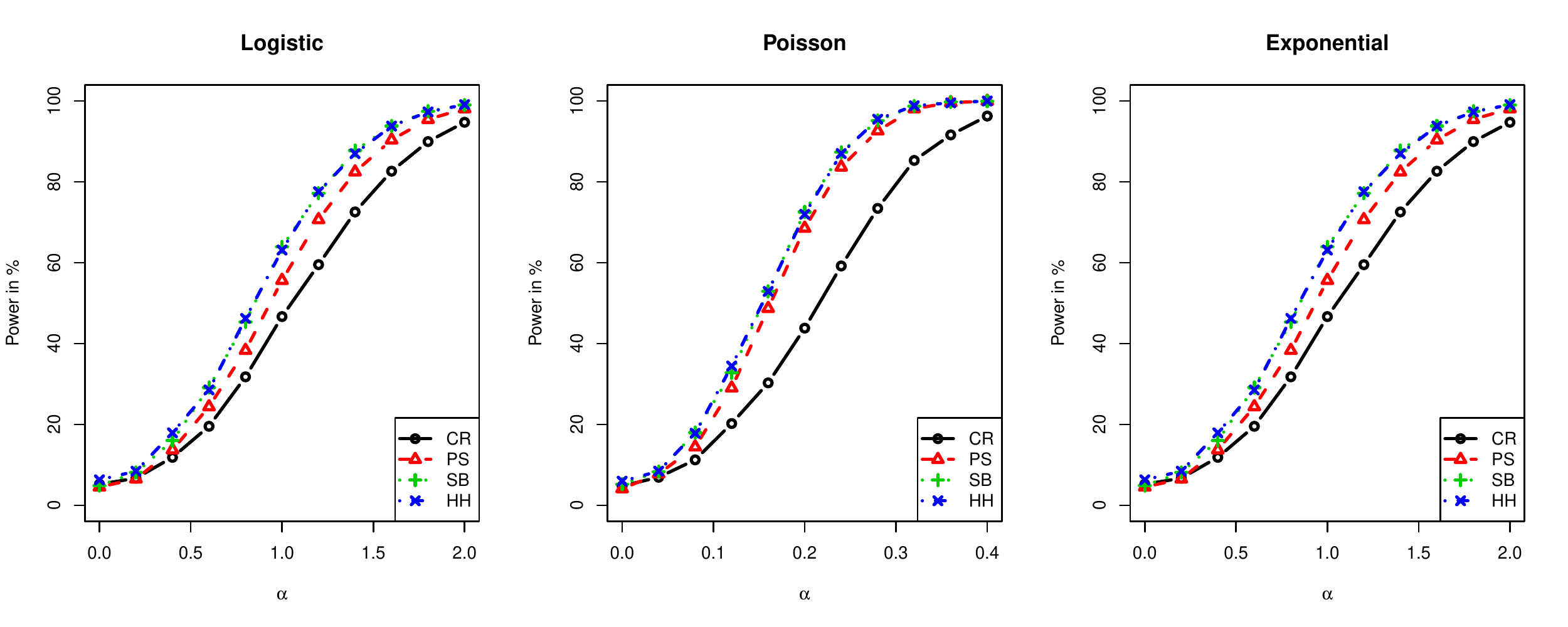}
\caption{{ 
			{\color{black} Simulated power of adjusted test for generalized linear models under various randomization methods (CR, PS, SB, HH) and sample size ($N=200$). }
			N, sample size; CR, complete randomization; PS, Pocock and Simon's 
	minimization;
	SB, stratified permuted block design; HH, the method proposed by Hu and Hu \cite{Hu2012}.}}
\end{figure}

\clearpage
\newpage

\section*{Appendix A: Proofs}

We only prove Theorem \ref{thm:TestGeneral-NC} for a covariate-adaptive design that satisfies Conditions (A)-(B) under generalized linear models with non-canonical link functions; an analogous (and simpler) argument establishes Theorems \ref{thm:TestWithinStratumBounded},  \ref{thm:TestGeneral}, and \ref{thm:TestWithinStratumBounded-NC}. We first prove the following lemmas.
For notational simplicity, $\nu_i=h(\mu+\beta x_i)$ is used throughout the proof section.

\begin{lemma}\label{lem:lem1}
	Suppose that Conditions (A)-(B) are satisfied; then under the null hypothesis 
	$H_0: \delta=0$, we have, as $n \rightarrow \infty$,
	\begin{align*}
	\left(\frac{1}{\sqrt{n}}\sum_{i=1}^{n}(2T_i-1)(y_i-\nu_i),
	\frac{1}{\sqrt{n}}\sum_{i=1}^{n}(2T_i-1)\nu_i\right) \Dcon
	\left(\xi_1,\xi_2\right),
	\end{align*}
	where  $\xi_1$ and $\xi_2$ are independent, and $\xi_1 \sim N(0, 
	\Ex[\var(Y|X)])$ and $\xi_2 \sim N(0,\sigma_h^2)$.
\end{lemma}

\begin{proof}[Proof]
	Let $\mathcal{X}=\{x_1,\ldots,x_n\}$ and $\mathcal{T}=\{T_1,\ldots,T_n\}$.
	As $T_i$ and $(y_i-\nu_i)$ are independent given $\mathcal{X}$ under the 
	null hypothesis, then
	\begin{align*}
	\Ex[\frac{1}{n}\sum (2T_i-1)(y_i-\nu_i)|\mathcal{X}]=\frac{1}{n}\sum 
	\Ex[(2T_i-1)|\mathcal{X}]\Ex[(y_i-\nu_i)|\mathcal{X}]=0,
	\end{align*}
	and using the fact $(2T_i-1)^2=1$,
	\begin{align*}
	\var[\frac{1}{n}\sum (2T_i-1)(y_i-\nu_i)|\mathcal{X}]=\frac{\sum 
		\Ex[(2T_i-1)^2 (y_i-\nu_i)^2|\mathcal{X}]}{n^2}=\frac{\sum 
		\var(y_i|x_i)}{n^2}.
	\end{align*}
	Therefore, by the central limit theorem, the conditional distribution of 
	\begin{align*}
	\frac{1}{\sqrt n}\sum (2T_i-1)(y_i-\nu_i)
	\end{align*}
	given $(\mathcal{X}, \mathcal{T})$ is asymptotically normal with mean zero and 
	variance $n^{-1}\sum \var(y_i|x_i)$, which converges to $\Ex[\var(Y|X)]$ by the 
	law of large numbers.
	
	We further let $A_n=n^{-1/2}\sum (2T_i-1)(y_i-\nu_i)$ and $B_n=n^{-1/2}\sum (2T_i-1)\nu_i$,
	then by the independence of $\sum (2T_i-1)(y_i-\nu_i)$ and $\sum (2T_i-1)\nu_i$ 
	given $(\mathcal{X}, \mathcal{T})$, we have, for any real numbers $a$ and $b$,
	\begin{align*}
	P\{A_n\le a,B_n\le b\}&=\Ex[P\{A_n\le a,B_n\le b|\mathcal{X}, \mathcal{T}\}]\\
	&=\Ex[P\{A_n\le a|\mathcal{X}, \mathcal{T}\}P\{B_n\le b|\mathcal{X}, 
	\mathcal{T}\}]\\
	&=\Ex[P\{A_n\le a|\mathcal{X}, \mathcal{T}\}I\{B_n\le b\}]\\
	&=\Ex[(P\{A_n\le a|\mathcal{X}, \mathcal{T}\}-P\{\xi_1\le a\})I\{B_n\le b\}]
	+P(\xi_1\le a)P(B_n\le b)\\
	&\rightarrow P(\xi_1\le a)P(\xi_2\le b),
	\end{align*}
	where the convergence follows from the dominated convergence theorem and 
	Condition (B).
	Hence the desired conclusion is proved.
\end{proof}

\begin{lemma}\label{lem:lem2}
	Suppose that Conditions (A)-(B) are satisfied; then under the null hypothesis 
	$H_0: \delta=0$, we have, as $n \rightarrow \infty$,
	\begin{align*}
	\frac{1}{n_1}\sum_{i=1}^{n}T_iy_i \Pcon \Ex[Y],
	\end{align*}
	and
	\begin{align*}
	\frac{1}{n_0}\sum_{i=1}^{n}(1-T_i)y_i \Pcon \Ex[Y],
	\end{align*}
	where $n_1=\sum_{i=1}^{n} T_i$, $n_0=\sum_{i=1}^{n}(1-T_i)=n-n_1$.
\end{lemma}

\begin{proof}[\textbf{Proof}]
	We only show the first result, and the second result can be proved 
	similarly.	
	Note that
	\begin{align*}
	\frac{1}{n_1}\sum T_iy_i =
	\frac{n}{2n_1}\left[\frac{1}{n}\sum(2T_i-1)y_i+\frac{1}{n}\sum y_i\right].
	\end{align*}
	By Lemma \ref{lem:lem1}, we have
	\begin{align*}
	\frac{1}{n}\sum(2T_i-1)y_i\Pcon 0,
	\end{align*}
	which, together with Condition (A) and the law of large numbers, implies that
	\begin{align*}
	\frac{1}{n_1}T_iy_i \Pcon \Ex[Y].
	\end{align*}
\end{proof}

\begin{proof}[\textbf{Proof of Theorem \ref{thm:TestGeneral-NC}}]
	We first show that 
	\begin{align}\label{eq:appx-numerator}
	\hat{\delta}=\frac{1}{h'h^{-1}(\Ex[Y])}
	\frac{2\sum(2T_i-1)y_i}{n}+o_P(n^{-1/2}).
	\end{align}
	By Taylor expansion, we have
	\begin{align*}
	h^{-1}(\frac{\sum T_iy_i}{n_1})=h^{-1}(\Ex[Y])
	+h^{-1(1)}(\Ex[Y])(\frac{\sum T_iy_i}{n_1}-\Ex[Y])
	+\frac{1}{2}h^{-1(2)}(\xi)(\frac{\sum T_iy_i}{n_1}-\mathbb{E}y)^2,
	\end{align*}
	where $h^{-1(k)}$ denotes the $k$th derivative of $h^{-1}$, $k=1,2$, and $\xi$ 
	lies between ${\sum T_iy_i}/{n_1}$ and $\Ex[Y]$.
	
	Based on Lemma \ref{lem:lem2} and the fact that the overall imbalance is 
	bounded in probability, we have
	\begin{align*}
	\frac{\sum T_iy_i}{n_1}-\Ex[Y]&=\frac{2}{n}\sum T_iy_i-\Ex[Y]
	-\frac{\sum T_iy_i}{n_1}\frac{D_n}{n}\\
	&=\frac{2}{n}\sum T_iy_i-\Ex[Y]+o_P(n^{-1/2}).
	\end{align*}
	
	Furthermore, by Lemma \ref{lem:lem1} and the central limit theorem, we have
	\begin{align*}
	\sqrt{n}(\frac{\sum T_iy_i}{n_1}-\Ex[Y])=\frac{\sum(2T_i-1)y_i}{\sqrt{n}}
	+\frac{\sum y_i-n\Ex[Y]}{\sqrt{n}}+o_P(1)=O_P(1),
	\end{align*}
	which, together with Lemma \ref{lem:lem2}, implies that
	\begin{align*}
	h^{-1}(\frac{\sum T_iy_i}{n_1})=h^{-1}(\Ex[Y])
	+h^{-1(1)}(\Ex[Y])(\frac{2\sum T_iy_i}{n}-\Ex[Y])
	+o_P(n^{-1/2}),
	\end{align*}
	and, by symmetry, 
	\begin{align*}
	h^{-1}(\frac{\sum (1-T_i)y_i}{n_0})=h^{-1}(\Ex[Y])
	+h^{-1(1)}(\Ex[Y])(\frac{2\sum (1-T_i)y_i}{n}-\Ex[Y])
	+o_P(n^{-1/2}).
	\end{align*}
	
	Therefore, \eqref{eq:appx-numerator} is proved by noting that 
	$h^{-1(1)}(\Ex[Y])=1/h'h^{-1}(\Ex[Y])$.
	
	On the other side, the variances of $\hat\mu$ and $\hat\delta$ are estimated 
	based on the information matrix $I(\hat{\mu},\hat{\delta})$, where 
	$\hat{\eta}_i=\hat{\mu}+\hat{\delta}T_i$ and
	\begin{align*}
	I(\hat{\mu},\hat{\delta})=\frac{1}{\phi}
	\left[\begin{array}{cc}
	\sum h^{'}(\hat{\eta}_i)\gamma^{'}(\hat{\eta}_i) & \sum 
	T_ih^{'}(\hat{\eta}_i)\gamma^{'}(\hat{\eta}_i)\\
	\sum T_ih^{'}(\hat{\eta}_i)\gamma^{'}(\hat{\eta}_i)	& \sum 
	T_ih^{'}(\hat{\eta}_i)\gamma^{'}(\hat{\eta}_i)
	\end{array}\right].
	\end{align*}
	Notice that
	\begin{align*}
	\frac{1}{n}\sum h^{'}(\hat{\mu}+\hat{\delta} 
	T_i)\gamma^{'}(\hat{\mu}+\hat{\delta} 
	T_i)
	&=\frac{n_1}{n}h^{'}(\hat{\mu}+\hat{\delta})\gamma{'}(\hat{\mu}+\hat{\delta})
	+\frac{n_0}{n}h^{'}(\hat{\mu})\gamma{'}(\hat{\mu})\\
	&=\frac{n_1}{n}h^{'}h^{-1}(\frac{\sum T_iy_i}{n_1})\gamma^{'}h^{-1}(\frac{\sum 
		T_iy_i}{n_1})\\
	& +\frac{n_0}{n}h^{'}h^{-1}(\frac{\sum 
		(1-T_i)y_i}{n_0})\gamma^{'}h^{-1}(\frac{\sum (1-T_i)y_i}{n_0})\\
	&\Pcon h^{'}h^{-1}(\Ex[Y])\gamma^{'}h^{-1}(\Ex[Y]),
	\end{align*}
	and, similarly,
	\begin{align*}
	\frac{1}{n}\sum T_ih^{'}(\hat{\mu}+\hat{\delta} 
	T_i)\gamma^{'}(\hat{\mu}+\hat{\delta} 
	T_i)
	&=\frac{n_1}{n}h^{'}(\hat{\mu}+\hat{\delta})\gamma{'}(\hat{\mu}+\hat{\delta})\\
	&\Pcon \frac{1}{2}h^{'}h^{-1}(\Ex[Y])\gamma^{'}h^{-1}(\Ex[Y]).
	\end{align*}
	Then we have
	\begin{align*}
	\frac{1}{n}I(\hat{\mu},\hat{\delta})\Pcon
	\frac{h^{'}h^{-1}(\Ex[Y])\gamma^{'}h^{-1}(\Ex[Y])}{\phi}
	\left[\begin{array}{cc}
	1	& 1/2\\
	1/2	& 1/2
	\end{array}\right],
	\end{align*}
	which follows that 
	\begin{align*}
	nI^{-1}(\hat{\mu},\hat{\delta})\Pcon
	\frac{\phi}{h^{'}h^{-1}(\Ex[Y])\gamma^{'}h^{-1}(\Ex[Y])}
	\left[\begin{array}{cc}
	2	& -2\\
	-2	& 4
	\end{array}\right].
	\end{align*}
	As a result,
	\begin{align}\label{eq:appx-denominator}
	\sqrt{n}{\widehat{\se}(\hat\delta)}\Pcon 
	\left\{\frac{4\phi}{h^{'}h^{-1}(\Ex[Y])\gamma^{'}h^{-1}(\Ex[Y])}\right\}^
	{\frac{1}{2}}
	\end{align}
	
	Combining \eqref{eq:appx-numerator} and \eqref{eq:appx-denominator}, we have
	\begin{align}\label{prop:delta-test}
	S=\frac{\hat{\delta}}{\widehat{\se}(\hat\delta)}
	=\frac{\sum (2T_i-1)y_i/\sqrt{n}}{\sqrt{\phi 
			h^{'}h^{-1}(\Ex[Y])/\gamma^{'}h^{-1}(\Ex[Y])}}+o_P(1),
	\end{align}
	and the asymptotic distribution of $S$ then follows from Lemma \ref{lem:lem1} 
	and the continuous mapping theorem.	
\end{proof}

\section*{Appendix B: Additional Simulation Results}
All of the simulation results in the main text are based on models with 2$\times$2 strata.
In this section, we consider the test for treatment effect with a large number of 
strata present. 
For this, we simulate four discrete covariates, with 2, 2, 3, and 4 levels, 
respectively, resulting in 48 strata.
The true model is assumed as follows
\begin{align}
\Ex(y_i|T_i,x_i)=h(\mu+\delta T_i+\beta_1 x_{i,1}+\beta_2 x_{i,2}
+\beta_3^t x_{i,3}+\beta_4^t x_{i,4}),
\end{align}  
where $x_{i,1}$ and $x_{i,2}$ are binary covariates with the probability of 0.5 
to take 0 or 1, $x_{i,3}^t$ is a discrete covariate with three equally possible 
values coded as $(0,0)$, $(1,0)$, and $(0,1)$, and $x_{i,4}^t$ is a discrete 
covariate with four equally possible values coded as $(0,0,0)$, $(1,0,0)$, 
$(0,1,0)$ and $(0,0,1)$.
The covariates are assumed to be independent of each other within each patient and between patients.
We still consider the three types of generalized linear models in Section
\ref{sec:Simulation-Size}. 
The model parameters are $\mu=-1$, $\beta_1=1$, 
$\beta_2=-1$, $\beta_3^t=(-2,1)$ and $\beta_4^t=(1,2,3)$ for logistic 
regression, and $\mu=0.8$, $\beta_1=0.25$, $\beta_2=-0.25$, 
$\beta_3^t=(-0.5,0.25)$, and $\beta_4^t=(0.25,0.5,0.75)$ for Poisson 
regression and the exponential model.

\begin{center} [Tables B1-B3 here.] \end{center}

In Tables B1-B3, the simulation results of the Wald test show patterns similar to 
those in Tables 1-3 with 2$\times$2 strata.
Under the three types of covariate-adaptive randomization, the Type I errors of 
Wald test are conservative for logistic regression, near or slightly larger 
than the nominal level for Poisson regression, and too anti-conservative for 
the exponential model.
After adjustment, the Type I errors of the Wald test are closer to the nominal 
level.
As the adjustment is based on the asymptotic results, the adjustment effects 
are less satisfactory when $N=200$, especially for the exponential model. 
As the sample size increases, the asymptotic properties begin to set in and the 
adjusted Type I errors are near 5\% when $N=$ 500 or 1000.
The power results are presented in Figure B1, from which it is clear that 
the adjusted tests are more powerful under covariate-adaptive designs than  under complete randomization.

\begin{center} [Figure B1 here.] \end{center}


\setcounter{table}{0}
\renewcommand{\thetable}{B\arabic{table}}

\setcounter{figure}{0}
\renewcommand{\thefigure}{B\arabic{figure}}

\newpage 
\begin{table}[htb]
	\begin{center}
	\caption{{\color{black} Simulated size in \% for logistic regression with 48 strata under various randomization methods (CR, PS, SB, HH), testing methods (Wald test, adjusted test), and sample sizes ($N=200, 500$). }}
		\bigskip
		\begin{tabular}{c c c c }
			$N$   & Randomization      & Wald Test &  Adjusted Test \\
			\hline
			\hline
			200   &   \textbf{CR}      & 4.76 & 5.14\\
			&   \textbf{PS}      & 1.80 & 3.08 \\
			&   \textbf{SB}      & 2.02 & 6.94 \\
			&   \textbf{HH}      & 1.92 & 5.56 \\
			\hline
			500   &   \textbf{CR}      & 4.86 & 4.90\\
			&   \textbf{PS}      & 1.68 & 3.60 \\
			&   \textbf{SB}      & 1.40 & 5.22 \\
			&   \textbf{HH}      & 1.56 & 5.24 \\
			\hline
			1000   &   \textbf{CR}      & 4.64 & 4.68\\
			&   \textbf{PS}      & 1.62 & 4.24 \\
			&   \textbf{SB}      & 1.50 & 5.22 \\
			&   \textbf{HH}      & 1.60 & 5.44 \\
			\hline
		\end{tabular}
	\end{center}
	Note: N, sample size; CR, complete randomization; PS, Pocock and Simon's 
	minimization;
	SB, stratified permuted block design; HH, the method proposed by Hu and Hu \cite{Hu2012}.
\end{table}

\begin{table}[htb]
	\begin{center}
	\caption{{\color{black} Simulated size in \% for Poisson regression with 48 strata under various randomization methods (CR, PS, SB, HH), testing methods (Wald test, adjusted test), and sample sizes ($N=200, 500$). }}
		\bigskip
		\begin{tabular}{c c c c }
			$N$   & Randomization      & Wald Test &  Adjusted Test \\
			\hline
			\hline
			200   &   \textbf{CR}      & 13.10 & 5.08\\
			&   \textbf{PS}      & 5.86 & 2.60 \\
			&   \textbf{SB}      & 6.82 & 6.94 \\
			&   \textbf{HH}      & 6.20 & 6.16 \\
			\hline
			500   &   \textbf{CR}      & 12.34 & 4.90\\
			&   \textbf{PS}      & 5.84 & 4.42 \\
			&   \textbf{SB}      & 6.12 & 5.98 \\
			&   \textbf{HH}      & 5.70 & 5.72 \\
			\hline
			1000   &   \textbf{CR}      & 12.62 & 5.14\\
			&   \textbf{PS}      & 6.10 & 5.00 \\
			&   \textbf{SB}      & 5.20 & 5.10 \\
			&   \textbf{HH}      & 4.54 & 4.54 \\
			\hline
		\end{tabular}
	\end{center}
	Note: N, sample size; CR, complete randomization; PS, Pocock and Simon's 
	minimization;
	SB, stratified permuted block design; HH, the method proposed by Hu and Hu \cite{Hu2012}.
\end{table}

\begin{table}[htb]
	\begin{center}
	\caption{{\color{black} Simulated size in \% for exponential model with 48 strata under various randomization methods (CR, PS, SB, HH), testing methods (Wald test, adjusted test), and sample sizes ($N=200, 500$). }}
		\bigskip
		\begin{tabular}{c c c c }
			$N$   & Randomization      & Wald Test &  Adjusted Test \\
			\hline
			\hline
			200   &   \textbf{CR}      & 43.94 & 6.00\\
			&   \textbf{PS}      & 40.88 & 3.38 \\
			&   \textbf{SB}      & 33.36 & 9.60 \\
			&   \textbf{HH}      & 34.66 & 10.12 \\
			\hline
			500   &   \textbf{CR}      & 45.84 & 5.08\\
			&   \textbf{PS}      & 44.14 & 4.52 \\
			&   \textbf{SB}      & 33.98 & 6.26 \\
			&   \textbf{HH}      & 34.46 & 6.08 \\
			\hline
			1000   &   \textbf{CR}      & 46.32 & 4.82\\
			&   \textbf{PS}      & 44.24 & 4.60 \\
			&   \textbf{SB}      & 35.50 & 5.62 \\
			&   \textbf{HH}      & 33.90 & 5.16 \\
			\hline
			
		\end{tabular}
	\end{center}
	Note: N, sample size; CR, complete randomization; PS, Pocock and Simon's 
	minimization;
	SB, stratified permuted block design; HH, the method proposed by Hu and Hu \cite{Hu2012}.
\end{table}

\begin{figure}
\centering
\includegraphics[width=15cm,height=6cm]{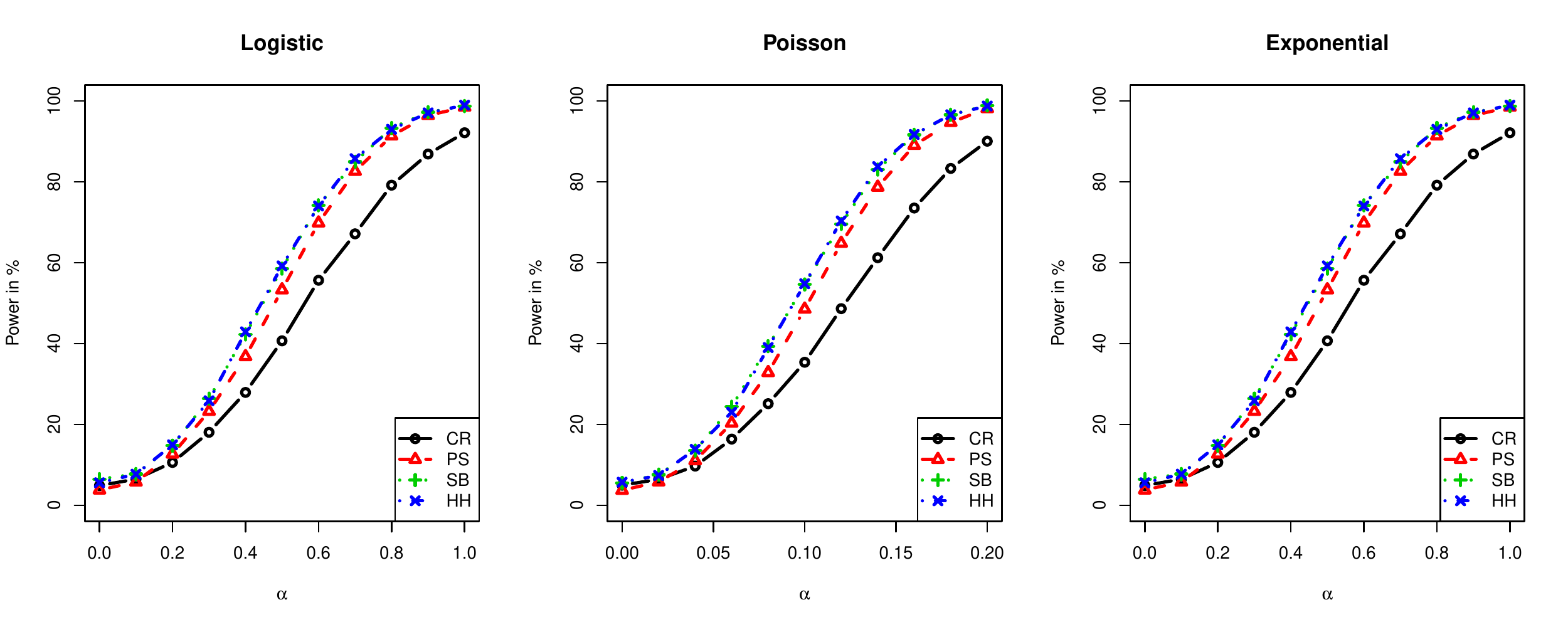}
\caption{{
			{\color{black} Simulated power of adjusted test for generalized linear models with 48 strata under various randomization methods (CR, PS, SB, HH) and sample size ($N=500$). }
			N, sample size; CR, complete randomization; PS, Pocock and Simon's 
	minimization;
	SB, stratified permuted block design; HH, the method proposed by Hu and Hu \cite{Hu2012}.}}
\end{figure}

\end{document}